\documentclass{article}

\usepackage{arxiv}
\usepackage{hyperref}       

\usepackage{amsmath,amsthm,amssymb,amsfonts}
\usepackage{mathtools}
\usepackage{algpseudocode}
\usepackage{algorithm}
\usepackage{graphicx}
\usepackage{enumerate}
\usepackage{float}
\usepackage{subcaption}
\usepackage{setspace}
\usepackage{color}
\usepackage{cite}
\usepackage{romannum}
\usepackage{tabularx}
\usepackage{stackrel}
\usepackage{relsize}
\usepackage{comment} 
\usepackage{verbatim}
\DeclareMathOperator*{\argmax}{arg\,max}

\newcommand{\trps}{^{\small \mathsf{T}}}																	

\newcommand{\diff}[1]{\text{d}{#1}}
\DeclarePairedDelimiterX{\norm}[1]{\lVert}{\rVert}{#1}

\renewcommand{\cite}[1]{[\citen{#1}]}

\newtheorem{theorem}{Theorem}[section]
\newtheorem{lemma}[theorem]{Lemma}

\newtheorem{remark}[theorem]{Remark}

\graphicspath{
{./figures/}
}

\title{\LARGE \bf
	Bounded-Rational Pursuit-Evasion Games
}

\author{
	Yue Guan%
	\thanks{PhD Candidate, School of Aerospace Engineering, Georgia Institute of Technology, Atlanta, GA, 30332-0150, USA. Email: {\tt\small yguan44@gatech.edu}}\\
	\And 
	Dipankar Maity%
	\thanks{Postdoctoral  Fellow, School of Aerospace Engineering, Georgia Institute of Technology, Atlanta, GA, 30332-0150, USA. Email: {\tt\small dmaity@gatech.edu}}\\
	\And
	Christopher M. Kroninger%
	\thanks{Research Scientist, Combat Capabilities Development Directorate,  Army Research Laboratory, Aberdeen Proving Ground, MD, 21005-5066, USA. Email:
		{\tt\small christopher.m.kroninger.civ@mail.mil}}\\
	\And
	Panagiotis Tsiotras%
	\thanks{David \& Andrew Lewis Chair and Professor, School of Aerospace Engineering,  Georgia Institute of Technology, Atlanta, GA 30332-0150, USA. Email: {\tt\small tsiotras@gatech.edu}, AIAA Fellow.}\\	
}

\begin{document}
	
	\maketitle
	
	\thispagestyle{empty}
	\pagestyle{empty}

	\begin{abstract} \label{sec:abstract}
	We present a framework that incorporates the idea of bounded rationality into dynamic stochastic pursuit-evasion games.
    The solution of a stochastic game is characterized, in general, by its (Nash) equilibria in feedback form.
	However, computing these Nash equilibrium strategies may require extensive computational resources.
    In this paper, the agents are modeled as bounded rational entities  having limited computational resources.
    We illustrate the framework by applying it to a pursuit-evasion game between two vehicles
    in a stochastic wind field, where both the pursuer and the evader are bounded rational.
    We show how such a game may be analyzed by properly casting it as an iterative sequence of finite-state Markov Decision Processes (MDPs).
    Leveraging tools and algorithms from cognitive hierarchy theory (``level-$k$ thinking") we compute the solution of the ensuing discrete game, while taking into consideration the rationality level of each agent.
    We also present an online algorithm for each agent to infer its opponent rationality level.
	\end{abstract}
	
	
	
	\section{Introduction}	\label{sec:intro}
	
	Pursuit-evasion games are a special class of dynamic games introduced in the 60s \cite{Issac:1965,willman1969formal,Ho:1965,Behn:1968}, which capture a wide range of diverse behaviors amongst dynamic adversarial agents.
	An extensive amount of literature exists on the topic;
	some notable examples include~\cite{Basar:1999,Balakrishnan:1973,Ho:1965,Hespanha:1999,Behn:1968,Bopardikar:2008,Shinar:1980,Lin:2015,vidal2002probabilistic}.
	However, most of these prior works \cite{Basar:1999,Ho:1965,Hespanha:2000,Shinar:1980,Willman:1969,Lin:2015} are concerned with finding the equilibrium strategy(s) of the game, which requires the assumption that all agents are rational~\cite{Myerson:1997}.
	The assumption of perfect rationality is considered to be too strong and perhaps unrealistic for many applications \cite{Simon:1985,Gilovich:2002}.
	Furthermore, finding the equilibrium of a game is, in general, computationally expensive~\cite{Daskalakis:2006,Starr:1969}.
	To address these issues, we study pursuit-evasion games (PEGs) with agents that are not perfectly rational, but rather \textit{bounded rational}.
	The idea of {bounded rationality} \cite{Kahneman:2003} was introduced in the economics community to address the issue
	that most agents do not adhere to the paradigm of maximum utility theory that assumes perfectly informed, rational agents.
	Several manifestations of bounded rationality exist in the literature, but in this paper we will adopt the one that makes use of
	the model of \textit{cognitive hierarchies} (CH) \cite{Ho:2004} and \textit{level-$k$ thinking} \cite{Ho:2013}.
	Under this notion, the agent no longer seeks the (Nash) equilibrium, but instead seeks a best response to some (believed or inferred)
	strategy of its opponent(s) according to an assumed level-$k$ responding structure.
	
	In this work, we formulate a PEG  involving bounded-rational agents in a stochastic environment (e.g., UAVs in a wind field).
	Specifically, we want to address how, in an air-combat scenario, an agent should optimally respond if its opponent(s) is not perfectly rational.
	 We start by considering a continuous-time PEG with continuous state-space, finite control/action set, perfect state observation, and stochastic unbounded noise.
	Building upon the Markov Chain approximation method~\cite{Kushner:1992}, we then construct a series of finite-state,
	discrete-time PEGs as an approximation to the original  continuous time, continuous state-space problem.
	We subsequently introduce the level-$k$ thinking framework to solve this discrete PEG in order
	to capture the bounded rational properties of the players.
	In the level-$k$ framework, each agent is assumed to best respond to the strategy of its opponent(s).
	This leads to an iterative decision process for both the pursue and the evader.
	The iterative process begins with ``level-$0$" pursuer(s) and evader(s), who have no prior assumptions about their opponents, or the structure of the game, and merely act randomly; i.e., they implement a random policy where each action is taken equally likely.
	The ``level-$k$" pursuer(s) and evader(s) presume that their opponent(s) are of ``level-($k$-1)", and they respond accordingly based on this assumption.
	The notion of rationality is captured through the parameter $k$; the higher the value of $k$, the higher the level of rationality of the agent and the higher this iterative process proceeds.
	Consequently, and in order to capture the notion of bounded rationality~\cite{Ho:2013,Li:2016}, each agent has a maximum level beyond which it cannot operate (different agents may have different maximum levels).
	Under this framework, the computational complexity of finding the optimal policy for a participating agent is significantly reduced~\cite{Li:2016}, since the game can be cast as sequence of single-sided Markov-Decision-Process (MDP) problems~\cite{Ho:2013} for each agent,
each one of which can be solved efficiently~\cite{Filar:1996}.
	
	One challenge in this framework is that in order to pick the optimal policies, each agent needs to know the rationality level of its opponent(s).
To perform such an inference, we  adopt an inference algorithm  based on the observed agent trajectories~\cite{Ho:2013}.
The algorithm assumes that both agents' rationality levels are fixed, and then lets the agents use a maximum likelihood estimator to update their beliefs regarding their opponent's rationality level.
We then generalize this inferring algorithm to the case where the opponent's perceived rationality level is fixed, while each agent adapts its rationality level based on its belief about its opponent's rationality level.

This work applies the framework of cognitive hierarchy \cite{Ho:2004} to PEGs and utilizes the concept of the agent's \textit{rationality level}.
The main contribution of this paper is the detailed derivation of a comprehensive and implementable method to find solutions to stochastic
PEGs involving bounded rational agents.
The proposed method first discretizes the continuous PEGs in a way that ensures the convergence of the optimal solution of the discretized game to that of the original continuous-state game.
In addition, we show that if level-$k$ thinking is applied to encode bounded rational decision-making of the agents,
the discrete PEGs can be solved efficiently using value iteration of the corresponding one-sided MDP problems for each agent.
The (optimal) policy adopted by an agent depends not only on its own (bounded) rationality level but also on
the rationality level of its opponent(s).
As a result, estimating or inferring the rationality level of the opponents is also of great interest.	
Therefore, we also propose an inferring algorithm that updates the agents' beliefs about the rationality level of their opponent(s) during the game.
The effectiveness of the proposed approach is demonstrated using a two-dimensional two-agent PEG in a stochastic wind field, in which the two agents' rationalities are bounded by different maximum levels.

The rest of the paper is organized as follows.
Section~\ref{sec:formulation} presents the formulation of the pursuit-evasion game in a stochastic environment and applies the Markov Chain approximation method to discretize the problem \cite{Kushner:1992}.
Section~\ref{sec:level-k} introduces the level-$k$ thinking framework and  Section~\ref{sec:value_Iteration} presents algorithmic methods for solving this problem.
This is followed by Section~\ref{sec:infer} that discusses the proposed inferring algorithms, and Section~\ref{sec:example} that presents an illustrative example of a pursuit-evasion game in a stochastic wind field, with two  bounded-rational agents.
Finally,  the paper is concluded in Section \ref{sec:conclusion} with a summary of the results and some possible extensions for future work.

	\section{Problem Formulation}\label{sec:formulation}
	
	We start by introducing a pursuit-evasion game which evolves in a stochastic environment (wind field) in continuous time.
	Under certain assumptions, the game can be discretized into a \textit{competitive Markov Decision Process} (cMDP).
	The connection between the cMDP and an associated standard one-sided MDP problem is addressed.

	\subsection{Continuous-Time PEG}\label{subsec:continous_PEG}
	
	Let us consider a two-agent pursuit-evasion differential game in which the pursuer and the evader are indexed by $i =1$ and $i=2$, respectively.
	In what follows, we will also use the index $-i$ to denote the opponent of agent $i$. As the names of the agents suggest, the \textit{pursuer} tries to \textit{capture} the evader, while the \textit{evader} tires to \textit{evade} the pursuer.
	Without loss of generality, we henceforth assume that $i=1$ is the pursuer and $i=2$ is the evader.
	For simplicity, we assume that the game evolves in a two-dimensional compact {domain} $C \subset \mathbb{R}^2$, and the position (state) of an agent at time $t$ is denoted by $p^i(t)=[p^i_x(t),p^i_y(t)]\trps\in \mathbb R^2$, for $i=1,2$.
	We define the state of the game as the joint positions of the two agents.
	Specifically, at a given time $t \ge 0$, the state of the game is represented as $s(t)=\big[p^1(t)\trps,p^2(t)\trps\big]\trps \in S = C\times C \subset \mathbb{R}^4$.
	
\subsubsection{Wind Field and Game Dynamics}	
	
	The game is played in the presence of a stochastic wind field, which will be modeled by a variant of a Wiener process.
	The only knowledge  the agents have regarding the stochastic wind field is the mean and standard deviation of the wind velocity.
	The dynamics of the agents are governed by stochastic differential equations (SDE) as follows
	\begin{align} \label{eqn:system_SDE}
	   \diff p^i=\begin{bmatrix} v^i\cos\theta^i+w_x(p^i) \\ v^i \sin{\theta^i} +w_y(p^i)
	    \end{bmatrix}\diff t+\diff W_t^i {(p^i)}, \quad i=1,2,
	\end{align}
	where $w_x(p^i)$ denotes the mean wind velocity in the $x$-direction at $p^i$, while $w_y$ denotes the mean wind velocity in the $y$-direction.
	In~(\ref{eqn:system_SDE}) $\diff W^i_t(p^i)$ is a two-dimensional Wiener process with the following properties:
	a) $\displaystyle \mathbb E[\diff W^i_t(p) \diff W^i_t(p)\trps]=\sigma^2_w(p) \mathbb{I}_2$
	\footnote{ For simplicity, in this paper we will assume $\sigma_w(p)=\sigma_w$ for all $p$.
	}	;
	b) At any location $p\in \mathbb R^2$, the Wiener increments are independent at different time instances, $\mathbb E[\diff W^i_t(p) \diff W^i_{\tau}(p)\trps]=0$ for all $\tau \ne t$;
	c) At any time $t$, the wind disturbance is spatially uncorrelated, i.e.,  $\mathbb E[\diff W^i_t(p) \diff W^i_t(q)\trps]=0$ for $p\ne q$
	\footnote{In case the wind disturbance is spatially correlated, then $\displaystyle \mathbb E[\diff W^i_t(p) \diff W^i_t(q)\trps] = h(p,q)$, where $h(p,p) = \sigma^2_w(p) \mathbb{I}_2$,  and $h(p,q) \ge 0$ for $p \ne q$.
	A popular choice for random field models in geostatistics is $ \displaystyle h(p,q)=\exp(-\gamma\|p-q\|)\mathbb{I}_2$ for some $\gamma>0$ \cite{isaaks1989applied}.}.
	For brevity, we will denote the joint action $(\theta ^1,\theta^2) \in \Theta$ by $\theta$, and write
	$\text{d}W^i_t$ instead of $\diff W^i_t(p^i)$.
	We will also write for brevity
	$b^i(p^i,\theta^i) = \left[
		v^i \cos{\theta^i} + w_x(p^i),~v^i \sin{\theta^i} + w_y(p^i)
		 \right]\trps$.
		
		 The action set for each agent contains a collection of a finite number of possible choices for the desired heading angles
		 $\theta^1$ and $\theta^2$.
	Specifically, it will be assumed that the agent $i$ moves along the direction of its heading angle $\theta ^i \in \Theta^i= \{0,\pi/2, \pi, 3\pi/2\}$ at a fixed speed $v^i \ge 0$.
	The joint action space for both the players is denoted as $\Theta=\Theta^1\times \Theta^2$.
	The dynamics of the state of the game  can then be re-written as
	\begin{equation}\label{eqn:SDE}
	\text{d}s =  \left[ {\begin{array}{c}
		b^1(p^1,\theta^1)	\\b^2(p^2,\theta^2)
		\end{array} } \right] \mathrm{d}t + \left[ {\begin{array}{c}
		\text{d}W^1_t	\\\text{d}W^2_t
		\end{array} } \right] = b(s,\theta)\, \mathrm{d}t +  \mathrm{d}W_t,
	\end{equation}
	where $b(s,\theta)=[b^1(p^1,\theta^1)\trps,b^2(p^2,\theta^2)\trps]\trps$, and $\diff W_t=[{\diff W^1_t}\trps, {\diff W^2_t}\trps]\trps$.
	
	The solution of the SDE in (\ref{eqn:SDE}) is a stochastic process $\{s(t);t\geq0\}$ such that
	\begin{align}\label{eqn:SDE_sol}
	s(t) =& s(0) + \int_{0}^{t} b(s(\tau),\theta(\tau)) \text{d} \tau + \int_{0}^{t}  \text{d} W_\tau\\
	\triangleq & ~~\Phi (s(0),t,\theta^{[0,t]}, \diff W^{[0,t]}),
	\end{align}
	where the last term in equation (\ref{eqn:SDE_sol}) is treated as the usual It\^{o} integral \cite{Oksendal:2003},
	$\theta^{[0,t]}=\{\theta(\tau); \tau \in [0,t]\}$ is the joint action history,
	and $\diff W^{[0,t]}$ is the realization of the wind field in the interval $[0,t]$.
	
	
	\subsubsection{Terminal Conditions of the Game: Crash, Capture and Evasion}
	
	We define three terminating conditions (the boundaries) for the PEG.
	\begin{itemize}
	\item  We first consider the scenario when agents may \textit{crash} into obstacles or may crash on the boundary of the set $C$.
	Let the closed set $O \subset C$ denote the obstacle-region in $C$.
	Then, agent $i$ crashes into an obstacle or the boundary of the set $\partial C$, if at some time $t>0$
	\begin{equation}\label{eqn:crash}
	\text{dist}\big(p^i(t), O \cup \partial C \big)=0,
	\end{equation}
	where the distance of a point $p$ from a set $M$ is $\text{dist}(p,M)\triangleq \inf_{m\in M}\|p-m\|_2$.
	We define two boundaries and their union in the state space corresponding to the cases where either agent crashes as
	\begin{align*}
	\partial S^i_{\text{crsh}} &= \Big\{s\in S: \text{dist}\big(p^i(t),O \cup \partial C \big)=0 \Big\}, \quad i=1,2\\
	\partial S_{\text{crsh}} &= \partial S^1_{\text{crsh}} \cup \partial S^2_{\text{crsh}}.
	\end{align*}
	
    \item	 \textit{Capture} is considered successful when the distance between the two agents at some time instance $t>0$ is less than a prescribed positive value $\rho$, while neither of the agent crashes. That is, for a successful capture at time $t>0$, the following two conditions are simultaneously satisfied
	\begin{equation}\label{eqn:capture}
	\begin{aligned}
	\text{dist}\big(p^1(t),p^2(t)\big)&\leq \rho,\\
	\text{dist}\big(p^i(t), \partial S_{\text{crsh}} \big)&>0,~~~\text{  for } i = 1,2.
	\end{aligned}
	\end{equation}
	The capture condition defines a boundary in the state space as follows:
	\begin{equation*}
	\partial S_{\text{cap}} = \Big\{s\in S: \text{dist}\big(p^1(t),p^2(t)\big) \leq \rho \Big\} \setminus  \partial S_{\text{crsh}}.
	\end{equation*}
	
    \item	\noindent \textit{Evasion} is considered to be successful when the evader enters a closed, non-empty evading region $E \subset \mathbb{R}^2$ without being captured, while neither of the agents crashes. In other words, evasion is successful, if at some time $t>0$, the following conditions are simultaneously satisfied
	\begin{equation}\label{eqn:evasion}
	\begin{aligned}
	\text{dist}&\big(p^2(t),E\big)=0,\\
	\text{dist}&\big(p^1(t),p^2(t)\big)> \rho,\\
	\text{dist}&\big(p^i(t), S_{\text{crsh}} \big)>0,~~~\mathrm{for}~~i = 1,2.
	\end{aligned}
	\end{equation}
 Similarly, the boundary for evasion is defined as
	\begin{equation*}
	\partial S_{\text{evs}} = \Big\{s\in S:  \text{dist}\big(p^2(t),E\big)=0 \Big\} \setminus \big ( \partial S_{\text{crsh}} \cup \partial S_{\text{cap}} \big).
	\end{equation*}
	\end{itemize}
	The boundary of the state space $\partial S$ is defined by the three (disjoint) terminal conditions (\ref{eqn:crash}), (\ref{eqn:capture}) and (\ref{eqn:evasion}) as
	\begin{equation*}
	\partial S = \partial S_{\text{crsh}} \cup \partial S_{\text{cap}} \cup \partial S_{\text{evs}}.
	\end{equation*}
	When the process $\{s(t); t\ge 0\}$ hits the boundary $\partial S$, the game terminates and the outcome of the game is determined by which of the boundaries ($\partial S_{\text{crsh}}, \partial S_{\text{cap}}, \partial S_{\text{evs}}$) is reached by the process.
	
	To this end, we define the interior of the state space $S=C\times C$ of the PEG game as
	$
	S^o=S \setminus \partial S.
	$
	Without loss of generality, we assume that the game starts at some initial condition $s(0) \in S^o$.

	\subsubsection{Admissible Strategies and Rewards}
	An admissible policy (or strategy) for an agent at time $t$
	is a measurable mapping from the observation history ($\{s(\tau); \tau\in [0,t]\}$) to a probability distribution over its action set.
	It is well known that, with 
	full state information, the best Markov policy performs as well as the best admissible strategy \cite{Oksendal:2003,Karatzas:2014}.
	A Markov policy depends only on the current state $s(t)$ and not on the entire history $\{s(\tau); \tau\in [0,t]\}$.
	The Markov policy  of agent $i$ is represented by the measurable function $\mu^i(\cdot,\cdot) : S^o \times \Theta^i\rightarrow \mathbb{R}_{+}$, which denotes the probability distribution that the agent $i$ places over the action space $\Theta^i$ at some state $s$, i.e., $\mu^i(s,\theta^i)$ denotes the probability of choosing action $\theta^i\in \Theta^i$ at state $s$.
	Consequently, for all $s \in S^o$, $\sum_{\theta^i \in \Theta^i} \mu^i(s,\theta^i) = 1$.
	The set of all such policies  for agent $i$ will be denoted by $\Pi^i$.
	Similarly, we denote the joint policy  $(\mu^1,\mu^2)$ of both agents as $\mu$ and the set of all such policies by $\Pi$.
	Let us define the \textit{first exit time} $T_\mu$ under the joint policy $\mu$ as
	\begin{equation*}
	\begin{aligned}
	T_\mu = \inf \Big\{t:s(t) \notin S^o, &~s(t) = \Phi(s(0),t, \theta^{[0,t]},W^{[0,t]}), \\
	&\text{ where } \mathbb P(\theta^i(\tau)=a) = \mu^i\big(s(\tau),a \big)\text{ for } \tau \in [0,t]\Big\}.
	\end{aligned}
	\end{equation*}
	
	Therefore, $T_\mu$ is a random variable that reflects the first time a successful capture, or evasion, or crash occurs.
	Specifically, $T_\mu$ is the first time the process hits the boundary $\partial S$ of the state space $S$.
	
	We define the terminal reward for the pursuer at the terminal state $s(T_\mu) \in \partial S$ as
	\begin{equation}\label{eqn:continuous_reward}
	g^1(s(T_\mu)) =
	\begin{cases}
	~~1 ,			& \text{if } s(T_\mu) \in \partial S_{\text{cap}},\text{ i.e., successful capture occurs,} \\
	-1,             & \text{if } s(T_\mu) \in \partial S_{\text{evs}},\text{ i.e., successful evasion occurs,}\\
	-1, 			& \text{if } s(T_\mu) \in \partial S^1_{\text{crsh}}\setminus \partial S^2_{\text{crsh}},\text{ i.e., only the pursuer crashes,}\\
	~~1, 			& \text{if } s(T_\mu) \in \partial S^2_{\text{crsh}} \setminus \partial S^1_{\text{crsh}},\text{ i.e., only the evader crashes,}\\
	~~0,          	& \text{if } s(T_\mu) \in \partial S^1_{\text{crsh}} \cap \partial S^2_{\text{crsh}}, \text{ i.e., both agents crash}.\\
	\end{cases}\\
	\end{equation}
	In addition, we assume a zero-sum game formulation, and we define the terminal reward for the evader by
	\begin{equation}
	g^2(s(T_\mu)) = -g^1(s(T_\mu)).
	\end{equation}
	
	{This work formulates the PEG as a \textit {game of type}, where we use the win-rate as the performance index of the agents, instead of other metrics such as capture time.
	As a result, {only} the terminal reward is included and no running reward is introduced.
	An extension to include running reward (such as fuel cost or time) into the formulation requires a minor modification of the proposed framework and will be addressed elsewhere.}

	 The expected reward-to-go function for agent $i$ under the joint policy $\mu = (\mu^1,\mu^2)\in \Pi$
	and initial state $s_0$ is $J^{i}_{\mu}$, defined as
	\begin{equation} \label{eq:costJ}
	J^{i}_{\mu} (s_0) = J^{i}_{\mu^1,\mu^2} (s_0)= \mathbb{E}_{s_0,\mu} \Big[g^i\big(s(T_\mu)\big)\Big], \qquad
	i = 1,2,
	\end{equation}
	where the conditional expectation is given by
	\begin{equation*}
	\mathbb{E}_{s_0,\mu} \big[\cdot\big] = \mathbb{E}\Big[\cdot \vert s(t) = \Phi(s_0, t, \theta^{[0,t]},W^{[0,t]}), \mathbb P( \theta^i(\tau)) = \mu^i(s(\tau),\theta^i(\tau)) \text{ for all } \tau \in [0,t] \Big].
	\end{equation*}
	In the PEG setting,  $g^i(\cdot)$ incorporates the reward for either a successful capture, evasion, or crash.
	Both agents try to maximize their own expected reward functions by choosing their respective policies $\mu^i$. Specifically, for a given initial state $s_0$, the optimization problem can be written as
	\begin{equation*}
	\begin{aligned}
	\sup_{\mu^1\in \Pi^1}& J^1_{\mu^1,\mu^2}(s_0),\\
	\sup_{\mu^2\in \Pi^2} &J^2_{\mu^1,\mu^2}(s_0).
	\end{aligned}
	\end{equation*}
	Since the game is zero-sum, the two reward-to-go functions satisfy $J^1_{\mu^1,\mu^2}(s_0) = -J^2_{\mu^1,\mu^2}(s_0)$ for all joint policies $(\mu^1,\mu^2) \in \Pi$ and all states $s_0 \in S^o$. As a result, by defining $J_{\mu^1,\mu^2}(s_0) =J^1_{\mu^1,\mu^2}(s_0)$, the previous optimization problem can be rewritten as
	\begin{equation}
	\sup_{\mu^1\in \Pi^1} \inf_{\mu^2\in \Pi^2}  J_{\mu^1,\mu^2}(s_0) =
	\inf_{\mu^2\in \Pi^2} \sup_{\mu^1\in \Pi^1} J_{\mu^1,\mu^2}(s_0).
	\end{equation}
	Having defined the continuous PEG, in the next section, we will discuss the discretization of this PEG such that the optimal solution of the discretized PEG converges to the optimal solution of the continuous PEG as the discretization step is reduced to zero.
	
	\subsection{Markov Chain Approximation}\label{sec:MCAM}
	
	In this section, we introduce the \textit{Markov Chain Approximation Method} (MCAM) \cite{Kushner:1992}, which will be used to discretize the controlled stochastic process in (\ref{eqn:SDE}) into a discrete-state, discrete-time {competitive Markov Decision Process} \cite{Filar:1996}.
	In the specific case of a pursuit-evasion game, we refer to the discretized process as the discrete Pursuit-Evasion Game (dPEG).
	
	A two-agent zero-sum cMDP (equivalently, dPEG) is a tuple $\mathcal{M} = (S, \Theta^1, \Theta^2, P ,G)$ where $S$ is the finite set of states, $\Theta^i$ is the action set of agent $i$,
	and $P(\cdot|\cdot,\cdot, \cdot):S \times S \times\Theta^1 \times \Theta^2 \rightarrow \mathbb{R}_{\geq0}$ is a function that denotes the transition probabilities satisfying $\sum_{s' \in S}P(s'|s,\theta^1,\theta^2) = 1$ for all $s \in S$ and all $\theta^i \in \Theta^i, ~i = 1,2$.
	The mapping $G: S \rightarrow \mathbb{R}$ is the terminal reward function.
	Without loss of generality, we assume the pursuer (agent 1) maximizes $G$, while the evader (agent 2) minimizes $G$.
	Suppose we start at time 0 at an initial state $s_0 \in S$.
	At time $n \geq 0$, the current state is $s_n$, and we apply a joint control $\theta_n \in \Theta$ to arrive at the next state $s_{n+1}$ according to the transition probability function $P$.
	The result is a controlled Markov Chain $\{s_n; n\in \mathbb{N}_0=\mathbb{N}\cup\{0\}\}$ \cite{Kushner:1992}.
	Notice that such a chain depends on the choice of the initial state $s_0$ and the control sequence picked by both agents $\{\theta^i_n: n\in \mathbb{N}_0,~i=1,2\}$ .
	Up to any time $N > 0$, the observed chain is denoted as $s^{[0,N]} = \{s_n;n \leq N\}$ and is generated by the joint action sequence $\theta^{[0,N-1]}=\{\theta_n, n\leq N-1\}$.
	The chain $s^{[0,N]}$, the control sequence $\theta^{[0,N-1]}$, and the initial state $s_0$ are also referred to as the \textit{trajectory} of $\mathcal{M}$ under the given
	sequence of controls and the initial state.
	In the following, we present a method (Markov Chain Approximation) to obtain an equivalent cMDP from the PEG defined in Section \ref{subsec:continous_PEG}.

  Given a PEG defined by the dynamics (\ref{eqn:SDE}) with the terminal rewards as in (\ref{eqn:continuous_reward}), we adopt a spatio-temporal discretization scheme to reduce the problem to a finite-state cMDP.
	 We start by discretizing the compact region $C$ into a grid-like environment with grid-size $h>0$ (see Fig.~\ref{fig:MCAM}).
	The discretization size $h$ is a user-defined parameter that determines the coarseness of the discretization.
		 In Fig.~\ref{fig:MCAM}, the shaded red and gray areas correspond to evasion areas and obstacles respectively, and the orange and blue markers represent the 
	 positions of the pursuer and the evader.	
	Let us define the set $C_h=\{c_{h,1},\ldots,c_{h,N}\}$ denoting the elementary squares of the grid, so that we can always properly cover the compact set $C$ with $C_h$, namely, $C \subseteq C_h$.
	The cardinality of $C_h$, which is $N$, depends on the grid-size $h$.
	However, due to the compactness of $C$, the cardinality of $C_h$ is always finite for any $h$.
	Note that each $c_{h,j}\in C_h$ represents a unique square of size $h$ in $\mathbb R^2$.
	Let $c^o_{h,j}$ be the interior of the square cell represented by $c_{h,j}$.
	We extend the notion of obstacles and boundaries from $C$ to $C_h$ in the following way.
	The cell $c_{h,j}$ is labeled as an \textit{obstacle cell} if and only if $c_{h,j}^o \cap O \ne \emptyset.$
	The cell $c_{h,j}$ is labeled as a \textit{boundary cell} if and only if  $c_{h,j}^o \cap \partial C \ne \emptyset$. 
	Similarly, if $c_{h,j}^o\cap E\ne \emptyset$, then $c_{h,j}$ is labeled as an \textit{evasion cell}.
	We will say that agent $i$ is in cell $c_{h,j}$ if and only if $p^i(t) \in c_{h,j}^o$.
	If an agent is located on the common boundary of two cells (for example $p^i(t) \in c_{h,j} \cap c_{h,k}$), it is assigned to one of the cells ($j$ or $k$) based on a prior assignment rule.
	Having designed $C_h$, we define the discretization in $S$ as $S_h=C_h\times C_h$, i.e., each element of $S_h$ represents the discretized joint-state of the players.
	Each $s_h\in S_h$ denotes a hypercube of side length $h$ in $\mathbb{R}^4$.
	Analogous to $c_{h,j}^o$, we define the interior of the hypercube $s_h$ by $s_h^o$.
	We now define
	\begin{equation}\label{eqn:boundary}
    	\begin{aligned}
    	\partial S^i_{h,\text{crsh}} &=\{s_h\in S_h: s_h^o\cap \partial S^i_{\text{crsh}} \ne \emptyset\},\\
    	\partial S_{h,\text{crsh}}~&=~\partial S^1_{h,\text{crsh}} \cup \partial S^2_{h,\text{crsh}},\\
    	\partial S_{h,\text{cap}}~&=\{s_h\in S_h\setminus \partial S_{\text{h,crsh}} : s_h^o \cap \partial S_{\text{cap}}\ne \emptyset\},\\
    	\partial S_{h,\text{evs}}~&= \{s_h\in S_h\setminus (\partial S_{\text{h,crsh}}\cup \partial S_{\text{h,cap}} ) : s_h^o \cap \partial S_{\text{evs}}\ne \emptyset\},
    	\end{aligned}
	\end{equation}
	and finally $\partial S_h=\partial S_{h,\text{crsh}} \cup \partial S_{h,\text{cap}} \cup \partial S_{h,\text{evs}}$.
	
	\begin{figure}[htb]
		\centering
		\includegraphics[width=0.9\linewidth]{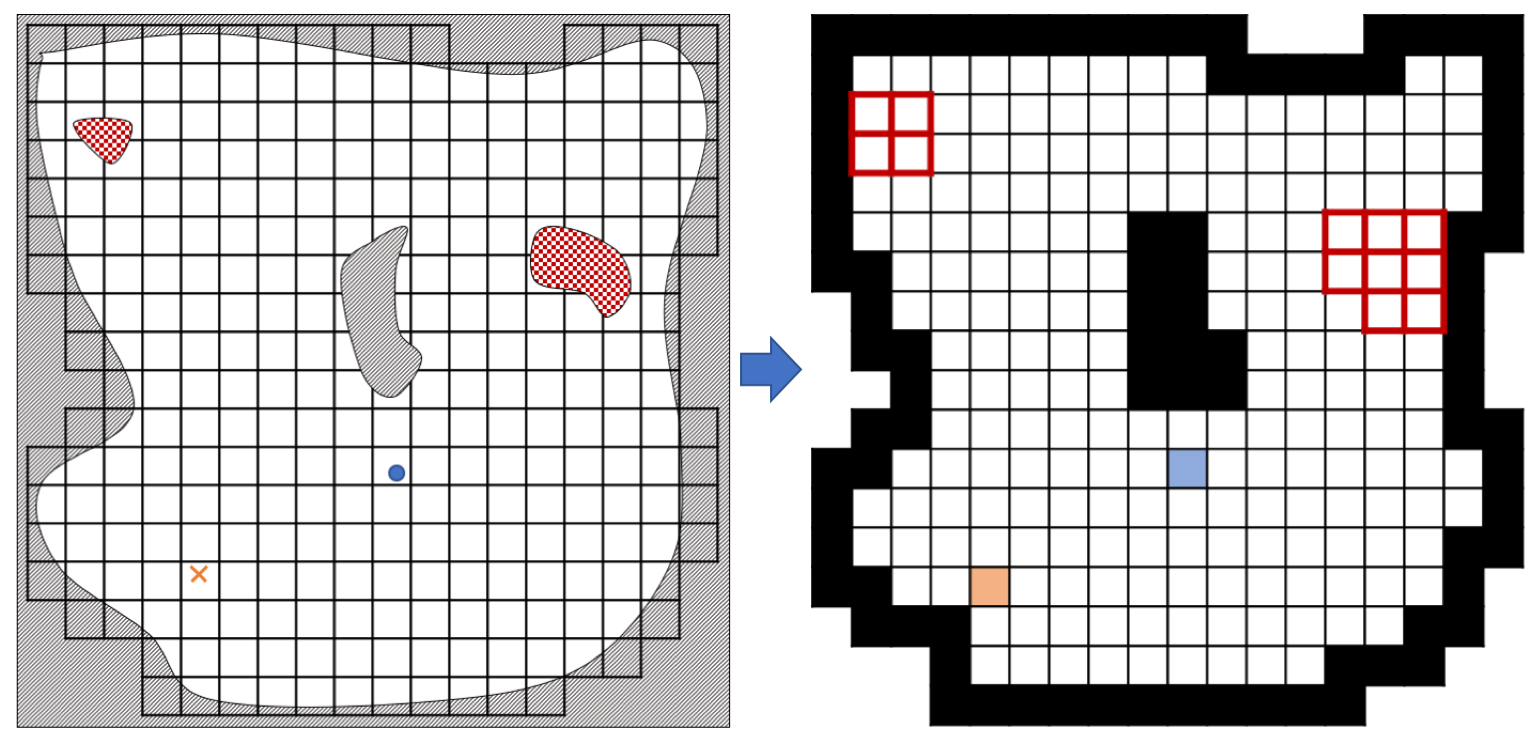}
		\caption{Example of a discretization of the two-dimensional space.
		On the left is the continuous space, and on 
		the right is the discretized space.}
		\label{fig:MCAM}
	\end{figure}
	
	The Markov Chain approximation method approximates the continuous zero-sum PEG using a sequence of cMDPs $\{\mathcal{M}_h\}$,
	where $\mathcal{M}_h = (S_h, \Theta^1, \Theta^2, P_h,G_h)$, and where $S_h=C_h\times C_h$ as defined earlier.
Note that	$\Theta^i$ is the original action set for $i = 1,2$.
	For each $h > 0$, let $s_h^{[0,N]} = \{s_h^n;n \leq N\}$ be a controlled Markov Chain on $\mathcal{M}_h$ until it hits $\partial S_h$. \\
	
	We associate each state $s \in S^o$ in the original continuous space with a non-negative interpolation interval $\Delta t_h(s)$, known as the \textit{holding time} \cite{Kushner:1992}.
	For each elementary hypercube $s_h \in S_h$, the centroid of $s_h$ is denoted as $\alpha(s_h)$, and satisfies
	\begin{align*}
	    \sup_{s\in s_h}\|\alpha(s_h)-s\|\le h/2.
	\end{align*}
	Let us also define $\Delta s_h^n = \alpha(s_h^{n+1})-\alpha(s_h^n)$.
	Since the mapping $\alpha$ is bijective, with a slight abuse of the notation, we use $s_h$ to denote both the cell and its centroid.
	For brevity, we denote $\Delta t^n_h$ to be the holding time at state $s^n_h$, i.e., $\Delta t_h(s^n_h)$.
	Furthermore, we also define $t^n_h = \sum_0^{n-1} \Delta t^n_h$ for $n \geq 1$ and $t^0_h = 0$.
	In addition, we define the discretized terminal cost $G_h$ in a similar manner as in (\ref{eqn:continuous_reward}):
	\begin{equation}\label{eqn:discretized_reward}
	\begin{aligned}
	G_h(s_h) &=
	\begin{cases}
	~~1 ,			& \text{if } s_h \in \partial S_{h,cap}, \\
	-1,             & \text{if } s_h \in \partial S_{h, evs},\\
	-1, 			& \text{if } s_h \in \partial S^1_{h, crsh} \setminus  \partial S^2_{h, crsh},\\
	~~1, 			& \text{if } s_h \in \partial S^2_{h, crsh} \setminus  \partial S^1_{h, crsh},\\
	~~0,          	& \text{if } s_h \in \partial S^1_{h, crsh} \cap  \partial S^2_{h, crsh} \text{ or } s_h \in S_h^o.\\
	\end{cases}\\
	\end{aligned}
	\end{equation}
	
	

\subsection{Discrete PEG}
	
	Let $\Omega_h$ be the sample space of $\mathcal{M}_h$,
	let $\theta_h^n = (\theta_h^{1,n},\theta_h^{2,n})$ be the joint action at time $n$ for the controlled Markov Chain.
	The holding times $\Delta t_h^n$ and the transition probabilities $P_h$ are chosen to satisfy \textit{the local consistency property} \cite{Kushner:1992}, with respect to (\ref{eqn:SDE}), which are given by the following conditions.
	\begin{enumerate}
		\item For all $s_h \in S_h$, $\lim_{h \rightarrow 0^+} \Delta t_h(s_h)=0$.
		\item For all $s_h \in S_h$ and all joint controls $\theta \in \Theta$:
		\begin{equation*}
		\begin{aligned}
		\lim_{h \rightarrow 0^+} \frac{\mathbb{E}_{P_h}[\Delta s_h^n \vert s_h^n = s, \theta_h^n = \theta]}{\Delta t_h(s)} &= b(s,\theta)\\
		\lim_{h \rightarrow 0^+} \frac{\text{Cov}_{P_h}[\Delta s_h^n \vert s_h^n = s, \theta_h^n = \theta]}{\Delta t_h(s)} &= \sigma_w \sigma_w \trps\\
		\lim_{h\rightarrow 0^+} ~\sup_{N\in \mathbb{N}_0, ~s_h^{[0,N]} \in \Omega_h} \norm{\Delta s_h^n}_2 &=0.	\end{aligned}
		\end{equation*}
	\end{enumerate}
	
	As the chain $\{s_h^n;n \in \mathbb{N}_0\}$ is a discrete-time process, we use an \textit{approximate continuous-time interpolation} \cite{Kushner:1992} to approximate the continuous-time process in (\ref{eqn:SDE}).
	We define the continuous-time interpolation $s_h(\cdot)$  of the chain $\{s_h^n;n \in \mathbb{N}_0\}$  and the continuous-time interpolation $\theta_h(\cdot)$ of the action sequence $\{\theta_h^n;n\in\mathbb{N}_0\}$ under the holding time function $\Delta t_h$ as follows: $s_h(\tau) = s_h^n$, and $\theta_h(\tau) = \theta_h^n$ for all $\tau \in [t^n_h,t^{n+1}_h)$.
	
	There are more than one ways to construct such a locally consistent Markov Chain \cite{Kushner:1992}.
	For our specific problem, we follow the construction found in \cite{Kushner:1992} that splits the control inputs from agent 1 (pursuer) and agent 2 (evader). \\
	Specifically, given the dynamics of the PEG as in (\ref{eqn:SDE}), we rewrite the drift term as
	\begin{equation*}
	b(s,\theta) = [b_1(s,\theta), ~b_2(s,\theta), ~b_3(s,\theta),~b_4(s,\theta)]\trps.
	\end{equation*}
Let us define the unit vectors in the four-dimensional state space as
	\begin{equation*}
	e_1 = [1 ~~0 ~~0 ~~0]\trps, ~~~ e_2 = [0 ~~1~~0 ~~0]\trps  , ~~~ e_3 = [0 ~~0 ~~1~~0 ]\trps  , ~~~ e_4 = [0 ~~0 ~~0~~1]\trps  .
	\end{equation*}
	We then define the quantity $Q_h(s)$ as
	\begin{equation*}
	Q_h(s) = h\max_{\theta \in \Theta}\{|b_1(s,\theta)|+|b_2(s,\theta)|+|b_3(s,\theta)|+|b_4(s,\theta)|\}+4~\sigma_w^2,
	\end{equation*}
	and define the interpolation interval as
	\begin{equation}\label{eqn:interp_time}
	\Delta t_h(s) = \frac{h^2}{Q_h(s)}.
	\end{equation}
	
	 Notice that $\Delta t_h(s)$ goes to zero as $h \rightarrow 0$, for all $s \in S_h$. The transition probabilities of the cMDP that approximates the original PEG can be calculated via
	\begin{equation}\label{eqn:trans_prob}
	\begin{split}
	P_h(s\pm he_j|s,\theta) & = (\frac{\sigma_w^2}{2} + h b_j^{\pm})/Q_h(s),\\
	P_h(s|s,\theta) &= 1 - \sum_{j=1}^{4} P_h(s+he_j|s,\theta)- \sum_{j=1}^{4} P_h(s-he_j|s,\theta), \\
	P_h(s'|s,\theta)& = 0,~ \text{if} ~s' \ne s\pm he_j  \text{ and } s' \ne s,
	\end{split}
	\end{equation}
	where $b^+ = \max \{b,0\}$ and $b^- = \max \{-b,0\}$, and, $s+he_j$ is the state $s'$ such that $\alpha(s')=\alpha(s)+e_j$.
	\noindent One can verify that the transition probabilities defined in (\ref{eqn:trans_prob}) are locally consistent with (\ref{eqn:SDE}). We will use the cMDP defined by $\mathcal{M}_h = (S_h,\Theta^1,\Theta^2,P_h,G_h)$ to approximate the original continuous-time PEG.

    Under the previous discretization scheme, any advantage an agent might have in terms of its speed is now translated into a probabilistic advantage,
    as shown in \eqref{eqn:trans_prob}.
    Specifically, one can verify from the construction of these probabilities ($P_h$ depends on $v^1, v^2$ through  $b_i(s,\theta)$) that with the same wind field present,
    the agent with a higher speed is more likely to end up in the cell that it intends to visit than the agent with a lower speed, as a direct consequence of (\ref{eqn:trans_prob}).
    Also, an agent with larger speed will be able to visit a larger number of neighboring cells.

	As stated in \cite{Kushner:1992}, local consistency implies the convergence of the continuous-time interpolations of the trajectories of the controlled Markov Chain to the trajectories of the stochastic dynamical system described in (\ref{eqn:SDE}), and the convergence of optimal reward-to-go functions of discrete cMDPs to that of the original problem.
	
	The definition of a policy for $\mathcal{M}_h$ is analogous to that defined in Section \ref{subsec:continous_PEG}.
	Similarly to the previous discussion, a (randomized) policy for agent $i$ is a mapping $\mu^i_h: S_h \times \Theta^i \rightarrow \mathbb{R}_{+}$, such that for all $s \in S_h$,  $\sum_{\theta^i \in \Theta^i} \mu_h^i(s,\theta^i) = 1$.
	The set of all such policies is denoted by $\Pi_h^i$.
    Let us use the notation $P_h(s'|s,\mu^1_h,\mu^2_h)$ to denote $P_h(s'|s,\mu^1_h,\mu^2_h)\triangleq \sum_{\theta^1\in \Theta^1,\theta^2\in \Theta^2}P_h(s'|s,\theta^1,\theta^2)\mu^1_h(s,\theta^1)\mu^2_h(s,\theta^2)$.
	To this end, given a joint policy $\mu_h = (\mu_h^1,\mu_h^2)$ and the initial state $s_0\in S_h$, the reward-to-go from $s_0$ due to $\mu_h$ is
	\begin{equation}\label{eqn:reward-to-go}
	J_{h,\mu_h} (s_0) =  \mathbb{E}_{P_h} \Big[G_h\big(s^{I_{h}}_h)\big)\Big],
	\end{equation}
	where the  $\mathbb{E}_{P_h} $ denotes the conditional expectation  under $P_h$ given $s^0_h=s_0$, and $\{s^n_h;n\in\mathbb{N}_0\}$ is the sequence of states of the controlled Markov Chain under the policy $\mu_h$, and $I_h$ is the termination time defined as $I_h = \min \{n:s_h^n \in \partial S_h\}$ \footnote{The termination time is well-defined, since the generated cMDP does not have any recurrent states.}.
	Notice that, although not expressed explicitly, the reward-to-go of the cMDP defined in \eqref{eqn:reward-to-go} is dependent on both agents' policies $\mu_h^i$~$(i=1,2)$ through $P_h$ and $s^{I_h}_h$.
	Therefore, the optimization problems for both agents are given by
	\begin{equation*}
	\sup_{\mu_h^1\in \Pi_h^1} \inf_{\mu_h^2\in \Pi_h^2}  J_{h,\mu_h^1,\mu_h^2}(s_0) =
	\inf_{\mu_h^2\in \Pi_h^2} \sup_{\mu_h^1\in \Pi_h^1} J_{h,\mu_h^1,\mu_h^2}(s_0).
	~~~~~~~~     \smash{\raisebox{\dimexpr.5\normalbaselineskip+.5\jot}{$
			$}}
	\end{equation*}
	
	 If the policy of the evader, $\mu_h^{2}$, is given, we can define the \textit{value function} or the \textit{optimal reward-to-go} for the pursuer (the maximizer), denoted by $V^1_{h,\mu^{2}_h},$ as
	\begin{equation}\label{eqn:value_function_p}
	V^1_{h,\mu^{2}_h}(s) = \sup_{\mu^1_h \in \Pi_h^1}J_{h,\mu^1_h,\mu^{2}_h}(s), ~ \text{ for } s \in S_h.
	\end{equation}
	Similarly, given the policy of the pursuer, $\mu_h^{1}$, the value function for the evader (the minimizer), denoted by $V^2_{h,\mu^{1}_h},$ satisfies
	\begin{equation}\label{eqn:value_function_e}
	V^2_{h,\mu^{1}_h}(s) = \inf_{\mu^2_h \in \Pi_h^2}J_{h,\mu^1_h,\mu^{2}_h}(s), ~ \text{ for } s \in S_h.
	\end{equation}
	In general, for a given opponent policy $\mu^{-i}_h$, an \textit{optimal policy}, if it exists, is denoted by $\mu^{i,*}_h$, and satisfies
	\begin{equation}
	J_{h,\mu^{i,*}_h,\mu^{-i}_h}(s) = V^i_{h,\mu^{-i}_h}(s) ~\text{ for all } s \in S_h.
	\end{equation}
	In later sections, we will show that if the opponent's policy is given, the optimal policy of the agent and the corresponding value function do exist for the discretized cMDP and we will also discuss the methods to find them.
	\vspace{5mm}
	
	\section{Level-k Thinking}\label{sec:level-k}
	
	As discussed in the introduction, under the framework of level-$k$ thinking, an agent best responds to a given policy of its opponent.
	Consider the cMDP $\mathcal{M}_h$ constructed as shown in the previous section.
	Let us  assume that the agent $i$ knows the policy of its opponent $\mu_h^{-i}$.
	Then, the transition probabilities in (\ref{eqn:trans_prob}), the expectation operator in (\ref{eqn:reward-to-go}), and the value functions in (\ref{eqn:value_function_p}) and (\ref{eqn:value_function_e}) will only depend on the agent $i$'s policy $\mu_h^{i}$.
	The cMDP $\mathcal{M}_h$ is now a standard  one-sided MDP.
	Similar to the previous definition of cMDP, given the opponent's policy $\mu^{-i}$, the MDP for agent $i$ is a tuple $\mathcal{M}^i_{h,\mu^-i} = (S_h, \Theta^i, P_h,G_h)$.
	The definition of each element in the tuple from Section~\ref{sec:formulation}.\ref{sec:MCAM} can be adapted to this context with minor changes. Given the policy of the evader $\mu^2$, let us consider the MDP for the pursuer (the maximizer), denoted by $\mathcal{M}^1_{h,\mu^2}$. To calculate the value function at a state $s \in S_h$, we solve the following Bellman equation:
	\begin{equation}\label{eqn:Bellman}
	\begin{aligned}
	V_{h,\mu_h^{2}}^1(s) &= \sup_{\mu^1 \in \Pi_h^1} \big\{  \mathbb{E}_{P_h} [V_{h,\mu_h^{2}}^1(s') \vert s,\mu_h^1]  \big\}\\
	& =  \sup_{\mu^1 \in \Pi_h^1} \Big\{  \sum_{s' \in S} P_h \big(s' \vert s,\mu_h^1, \mu_h^{2}\big)V_{h,\mu_h^{2}}^1(s')  \Big\}.
	\end{aligned}
	\end{equation}
	
	For the evader, the Bellman equation is similar but with an $\inf$ operator instead of a $\sup$ operator.
	We will discuss how to solve (\ref{eqn:Bellman}) numerically via value iteration in Section~\ref{sec:value_Iteration}.
	For now, let us assume that (\ref{eqn:Bellman}) can be solved by some algorithm
	given the opponents' policies $\mu_h^{-i}$, and the algorithm will find a unique value function $V_{h,\mu_h^{-i}}^i$ and an optimal policy $\mu_h^{i,*}$.
	
    In what follows, the superscript $k \in \{0,1,2,..\}$ within a parenthesis will denote the level of rationality.
	As we always consider the discrete-state MDP $\mathcal{M}_h$, we henceforth drop the subscript $h$ for notational brevity.
	
	The premise of level-$k$ thinking is that, if $\mu^{-i,(k)}$ is a level $k$ strategy for the opponent of agent $i$, we can define the level-$k$ MDP for agent $i$ as $\mathcal{M}^{i,(k)} = \mathcal{M}^i_{\mu^{-i,(k-1)}}$.
	The level $k+1$ strategy for agent $i$ at state $s$ is then defined as
	\begin{align}\label{eqn:level-k_MDP}
	    \mu^{i,(k+1)}(s)\in\argmax_{\mu^{i,(k+1)}\in \Pi^i}  \sum_{s' \in S} P \big(s' \vert s,\mu^{i,(k+1)}(s), \mu^{-i,(k)}(s)\big)V_{\mu^{-i,(k)}}^i(s'),
	\end{align}
	which is precisely the best response of agent $i$ to its opponent's strategy $\mu^{-i,(k)}$ (and hence it depends on $\mu^{-i,(k)}$).
	We therefore write (\ref{eqn:level-k_MDP}), equivalently, as
	\begin{equation}
	\mu^{i,(k+1)}   \in \text{BestResponse}(\mu^{-i,(k)}).
	\end{equation}
	Notice that if $\mu^{i,(k+1)}$ is the best response to $\mu^{-i,(k)}$, and $\mu^{-i,(k+2)}$ is the best response to $\mu^{i,(k+1)}$, then it is not necessarily  true that $\mu^{-i,(k)}=\mu^{-i,(k+2)}$.
Therefore, if agent $i$ starts with some strategy of level $k$, say $\mu^{i,(k)}$,
	then the agent might expect that its opponent will choose the best-response strategy $\mu^{-i,(k+1)}$.
	In fact, agent $i$ can compute $\mu^{-i,(k+1)}$ since it has knowledge of $P$ and $V^{-i}_{h,\mu^{i,(k)}}$.
	Having computed $\mu^{-i,(k+1)}$, agent $i$ will change its strategy to $\mu^{i,(k+2)}$, which is the best response to $\mu^{-i,(k+1)}$.
	
	\begin{lemma} \label{L:k-level}
	Given a two-player cMDP solved under the level-$k$ thinking framework, if $\mu^{i,(K+2)}=\mu^{i,(K)}$ for an agent $i$ at some finite level $K$, then the two agents reach a Nash equilibrium by applying the strategy pair $\left(\mu^{i,(K)},\mu^{-i,(K+1)}\right)$.
	\end{lemma}
	
	\begin{proof}
	By the construction of level-($k$) strategies in (\ref{eqn:level-k_MDP}), we know that
	\begin{align*}
	    \mu^{i,(K+2)}    &\in \text{BestResponse}(\mu^{-i,(K+1)}),\\
	    \mu^{-i,(K+1)}   &\in \text{BestResponse}(\mu^{-i,(K)}).
	\end{align*}
	Since $\mu^{i,(K+2)}=\mu^{i,(K)}$, we have the following fixed point property:
	\begin{align*}
	    \mu^{i,(K)}\quad  &\in \text{BestResponse}(\mu^{-i,(K+1)}),\\
	    \mu^{-i,(K+1)} &\in \text{BestResponse}(\mu^{-i,(K)}).
	\end{align*}
	The strategy pair $\left(\mu^{i,(K)},\mu^{-i,(K+1)}\right)$ then corresponds to a Nash equilibrium, by definition.
	\end{proof}
	
    \begin{remark}
        Suppose no constraint of maximum level is imposed, the level-$k$ iterative process terminates at some level $K$ if the condition $\mu^{i,(K+2)}=\mu^{i,(K)}$ is satisfied for some $i \in \{1,2\}$.
    \end{remark}

\begin{remark}
       Lemma~\ref{L:k-level} ensures that any converging level-$k$ strategy is a Nash Equilibrium.
       However,
        in general, there is no guarantee that the previous iterative level-$k$ construction will
  convergence.
      Furthermore, in case of multiple Nash equilibria, one cannot comment on which equilibrium the level-$k$ framework will converge.
    \end{remark}

	In practice, an agent with limited computational resources (e.g., a human) can continue this process only up to a certain level $k_{\text{max}}$.
	Clearly, before termination of the algorithm, the best response to a level-$k$ strategy $\mu^{-i,(k)}$ is, by definition, $\mu^{i,(k+1)}$ and not necessarily the NE strategy $\mu^{i,*}$.
	We want to emphasize the fact that the NE strategy $\mu^{i,*}$ is only optimal when the opponent plays
	$\mu^{-i,*}$.
	If agent $i$ is certain that its opponent is unable to compute $\mu^{-i,*}$ due to being bounded rational and plays strategy $\mu^{-i,(k)}$ for some $k \in \{0,\ldots, k_{\max}\}$, then there is an incentive for agent $i$ to play a strategy other than $\mu^{i,*}$ to maximize its own reward.
	Specifically, if it is the case that $-i$ will play its highest level strategy, i.e., $\mu^{-i,(k^{-i}_{\max})}$, then agent $i$ must play $\mu^{i,(k^{-i}_{\max}+1)}$, which is the best response to $\mu^{-i,(k^{-i}_{\max})}$ and not necessarily $\mu^{i,*}$.
	
    In the following, we will elaborate on how a level-$k$ strategy is constructed.
	For now, recall from the previous discussion that a level-$k$ strategy is built upon a level $k-1$ strategy of the opponent.
	To initialize the level-$k$ strategy construction, we first define the policies $\mu^{i,(0)}$ of the level-0 agents (we will refer these policies as the level-$0$ policies) as
	\begin{equation}
	\mu^{i,(0)}(\theta^{i},s) = \frac{1}{|\Theta^{i}|} ~~ \text{ for all } \theta^{i} \in \Theta^{i}, s \in S,
	\end{equation}
	where $|\Theta^{i}|$ denotes the cardinality of the action set $\Theta^{i}$ which is $\vert \Theta^i\vert =4$ in our case. We can view this level-0 policy as a uniform distribution over the action set $\Theta^{i}$. Namely, the level-0 pursuer and level-0 evader pick a heading from their action spaces uniformly.
	Such a choice reflects the idea that a level-$0$ agent is unable to  play any other action than one at random, and it does not care about optimality.
	The level-1 agents are given the level-0 policies, and they calculate their best response to their opponent's level-0 policy via (\ref{eqn:Bellman}).
	Specifically, $V_{\mu^{-i,(0)}}^{i,(1)}$, i.e.,
	the reward for agent $i$, given that its opponent uses stratregy
	$\mu^{-i,(0)}$, can be computed from the fixed point of the Bellman eqaution
	\begin{equation}
	V_{\mu^{-i,(0)}}^{i,(1)}(s) =  \sup_{\mu^{i,(1)} \in \Pi^i} \Big\{  \sum_{s' \in S} P \big(s' \vert s,\mu^{i,(1)}, \mu^{-i,(0)}\big)V_{\mu^{-i,(0)}}^{i,(1)}(s')  \Big\},
	\end{equation}
	with corresponding strategy $\mu^{i,(1)}$.
	For notational simplicity, we will drop the subscript $\mu^{-i,(k)}$ on the value functions and denote $ V_{\mu^{-i,(k-1)}}^{i,(k)}$ as $V^{i,(k)}$ instead.
	
	This process of building policies level-after-level continues until the prescribed maximum rationality level $k_{\max}^{i} \geq 1$ is achieved.
	In general, the level-$k$ policy for agent $i$ is calculated via
	\begin{equation}\label{eqn:level-k_Bellman}
	V^{i,(k^i)}(s) =  \sup_{\mu^{i,(k^i)} \in \Pi^i} \Bigg\{  \sum_{s' \in S} P \big(s' \vert s,\mu^{i,(k^i)}, \mu^{-i,(k^i-1)}\big)V^{i,(k^i)}(s')  \Bigg\},
	\end{equation}
	where $i\in\{1,2\}$ and $k^i \in \{1,2,\ldots,k^i_{\max}\}$.\\
	Note that, in order to calculate the policy $\mu^{i,(k)}$, for $k \in \{1,\ldots,k^i_{\max}	\}$ from (\ref{eqn:level-k_MDP}),
	agent $i$ one needs to calculate not only its own policy $\mu^{i,(k^i)}$ but also its opponent's policy $\mu^{-i,(k^i-1)}$, for all $k^i \in \{1,\ldots,k^i_{\max}\text{-}1\}$.
	For example, if the pursuer wants to calculate $\mu^{i,(k^i)}$ for $k^i=1,2,3$, it first needs to calculate the best response to $\mu^{2,(0)}$ as its level-1 pursuit policy $\mu^{1,(1)}$.
	To calculate $\mu^{1,(2)}$, which is the best response to $\mu^{2,(1)}$, the pursuer needs to calculate the level-1 evasion policy of the evader, $\mu^{2,(1)}$; similarly, to calculate $\mu^{1,(3)}$, it needs to calculate $\mu^{2,(2)}$ first, and so on.
	
	In summary, at level-$k$, we choose one optimizing agent in the group, while fixing all other agent's levels to $k-1$, and then let the chosen agent compute its best response.
	After all agents in the group have computed their best responses at level $k$, we move on to level $(k+1)$,
	and we iteratively build the structure of the rationality levels for each agent until the agent's maximum rationality level is reached.
	
	By assuming such a structure, we have reduced the cMDPs into a series of one-sided MDPs.
	For each one-sided MDP, $\mathcal{M}^{i,(k)}$
	the optimization is only over a single agent's policy and hence plenty of efficient numerical algorithms exist to solve for the optimal policy \cite{bertsekas1995dynamic}.
	In the following section, we will use the well-known value iteration algorithm to compute the best responses at each level $k$.
	\vspace{5mm}
	\section{Algorithmic Solutions}\label{sec:value_Iteration}
	
	At this point, notice that given the opponent's strategy $\mu^{-i,(k^i-1)}$, the optimal level $k^i$ strategy for agent $i$ can be found by an iterative method such as policy iteration or value iteration \cite{bertsekas1995dynamic}.
	For the sake of completeness, in this section we discuss such an iterative method to solve (\ref{eqn:level-k_Bellman}).
	In the following discussion, we will use the subscript $m$ to denote the iteration index.
	Following this notation, the value iterations for (\ref{eqn:level-k_Bellman}) are given by
	\begin{equation}
	\begin{aligned}
	V^{i,(k^i)}_{m+1}(s) &=  \sup_{\mu^{i,(k^i)} \in \Pi^i} \Bigg\{  \sum_{s' \in S} P \big(s' \vert s,\mu^{i,(k^i)}, \mu^{-i,(k^i-1)}\big)V_m^{i,(k^i)}(s')  \Bigg\}~&&\text{for } s\in S_h^o,\\
	V^{i,(k^i)}_{m}(s) & = G^i_h(s) && \text{for } s \in \partial S_h \text{ and for all } m.
	\end{aligned}
	\end{equation}
	It has been shown that there exist at least one pure policy \cite{Filar:1996} that solves (\ref{eqn:level-k_Bellman}).
	That is, there is an optimal policy $\mu^{i,(k^i)}$, such that at each state $s \in S_h$, there exists $\theta^i \in \Theta^i$ and $\mu^{i,(k^i)}(s,\theta^i)=1$.
	In other words, agent $i$ can always deterministically pick a heading angle at a given state $s$, while still attaining the optimal reward-to-go.
	As a result, we can define the class of pure policies by the map $\psi^{i}: S_h \times \Theta^{i} \rightarrow \{0,1\}$, where for all $s\in S_h$, one has
	\begin{equation}
	\sum_{\theta^i \in \Theta^i} \psi^{i}(s,\theta^i) = 1.
	\end{equation}
	Let the set $\Psi^{i}$ consisting of all such admissible pure policies.
	One immediately observes that $\Psi^i \subset \Pi^i$.
	We can significantly reduce the size of the search space in (\ref{eqn:level-k_Bellman}) by replacing a general mixed policy $\mu^i$ in with pure policy $\psi^i$, which yields
	\begin{equation}\label{eqn:value_iteration}
	\begin{aligned}
	V^{i,(k^i)}_{m+1}(s) &=  \max_{\psi^{i,(k^i)} \in \Psi^i} \Bigg\{  \sum_{s' \in S} P \big(s' \vert s,\psi^{i,(k^i)}, \mu^{-i,(k^i-1)}\big)V_m^{i,(k^i)}(s')  \Bigg\}~&&\text{for } s\in S^o_h, \\
	V^{i,(k^i)}_{m}(s) & = G^i_h(s) && \text{for } s \in \partial S_h \text{ and for all } m,
	\end{aligned}
	\end{equation}
	where $k^{i} \geq 1$.
	As the action space is assumed to be finite, solving the value iteration in (\ref{eqn:value_iteration}) at iteration $m$+1 becomes a problem of finding the heading angle at each state $s$ that maximizes the expected future rewards, given the value function $V_m^{i,(k^i)}$ from the previous iteration.
	With the policy of the opponent known, we can re-write (\ref{eqn:value_iteration}) as
	\begin{equation}\label{eqn:value_iteration_max_action}
	    \begin{aligned}
	    V^{i,(k^i)}_{m+1}(s) &=  \max_{\theta^i \in \Theta^i} \Bigg\{  \sum_{s' \in S} P_{\mu^{-i,(k^i-1)}} \big(s' \vert s,\theta^i \big)V_m^{i,(k^i)}(s')  \Bigg\}~&&\text{for } s\in S^o_h, \\
	    V^{i,(k^i)}_{m}(s) & = G^i_h(s) && \text{for } s \in \partial S_h \text{ and for all } m.
	    \end{aligned}
	\end{equation}
	Notice that the iteration in \eqref{eqn:value_iteration_max_action} does not have a discount factor due to the absence of a running reward in (\ref{eq:costJ}).
	The presence of a discount factor is essential to prove the convergence of value iteration in infinite horizon MDPs~\cite{Kushner:1962}.
However, with the transition probability of the controlled Markov Chain obtained from the Markov Chain approximation method, we have that no recurrent states exists for our problem
except from the absorbing boundary states (the terminal states) defined in (\ref{eqn:boundary}).
	It can then be shown that for this specific case the value functions $V_m^{i,(k^i)}$ in (\ref{eqn:value_iteration_max_action}) converge to the actual value function $V^{i,k^i}$ defined in \eqref{eqn:level-k_Bellman} as the iteration step $m\to \infty$~\cite{bertsekas1995dynamic}.

	\section{Inferring the Opponent's Level}\label{sec:infer}
	
	As discussed in the previous section, agent $i$ can compute its level-$k$ best response if it has knowledge of the opponent strategy $\mu^{-i,(k-1)}$.
	In Section \ref{sec:level-k}, we argued that agent $i$ can play its best strategy if it knows the level of the opponent is $k^{-i}$, and is able to compute $\mu^{i,(k^{-i}+1)}$.
	In real world situations, it is not always possible to know the opponent's rationality level exactly, and hence, the agent must infer the rationality level of its oppenent(s) based on the opponent(s) prior behavior.
	Below we propose an online algorithm to estimate the opponent's rationality level at each time based on the history of the game.
	The algorithm uses a maximum likelihood inferring algorithm utilizing the idea presented in~\cite{Yoshida:2008}.
	
	To this end, let us consider the one-sided MDP for agent $i$ at level-$k$ defined as $\mathcal{M}^{i,(k)}$.
	In a game setting without explicit communication between the agents, agent $i$ can only estimate its opponent's level by observing the trajectory of the states of the system.
	The proposed inferring algorithm provides a maximum likelihood estimator (MLE) of the opponent's level $k^{-i}$, given agent $i$'s level $k^{i}$ and an observed trajectory till the current time $s^{[0,N]} = \{s_{0},s_1, ..., s_N\}$.
	
	\subsection{Fixed Level Model}
	
	We first make the assumption that both agents set their levels at the beginning of the game and do not change their levels over subsequent stages.
	Suppose the maximum rationality level of agent $i$ is $k_{\max}^{i}$.
	Based on the discussion in the Section~\ref{sec:level-k}, agent $i$ should also have access to its opponent's policies till level-($k^i_\text{max}-1$).
	We denote the set of the opponent's level that agent $i$ has access to as $\mathcal{K}^{-i,(k^i_\text{max}\text{-}1)}= \{0,1,...,k_{\max}^{i}\text{-}1\}$.
	Given that agent $i$ plays at level $k^{i}$ and the game starts at initial state $s_{0}$, suppose that the agent $i$ observes a trajectory
	from stage $t_{0}$ to $t_N$ as $s^{[0,N]} = \{s_{0}, ..., s_N\}$.
	Assuming that the opponent plays at level $k^{-i} \in \mathcal{K}^{-i,(k^i_\text{max}\text{-}1)}$, the probability of observing the specific trajectory $s^{[0,N]}$ is
	\begin{equation}\label{eqn:infer_trajectory_conditional}
	\mathbb{P}(s^{[0,N]} |k^i,k^{-i})=\mathbb{P}(s^{[0,N]} |\mu^{i,(k^i)},\mu^{-i,(k^{-i})}) = \prod_{n=0}^{N-1}P_h\Big(s_{n+1}|s_n,~\mu^{i,(k^{i})},\mu^{-i,(k^{-i})}\Big).
	\end{equation}
	The estimate of the opponent's level can then be selected as the maximum likelihood estimator
	\begin{equation}
	\hat{k}^{-i} \in
	\argmax_{k^{-i} \in \mathcal{K}^{-i,(k^i_{\max}-1)}}
	\mathbb{P}\Big(s^{[0,N]} |k^{-i},k^{i}\Big).
	\end{equation}
	
	One can easily extend this algorithm to infer the opponent's level from a moving observation window of length $w+1$.
	Suppose a trajectory $s^{[N-w,N]}$ from stage $t_{N-w}$ to $t_{N}$ is observed, then the conditional probability in (\ref{eqn:infer_trajectory_conditional}) is changed to
    \begin{equation}
	\mathbb{P}(s^{[N-w,N]} |k^i,k^{-i})=\mathbb{P}(s^{[N-w,N]} |\mu^{i,(k^i)},\mu^{-i,(k^{-i})}) = \prod_{n=N-w}^{N-1}P_h\Big(s_{n+1}|s_n,~\mu^{i,(k^{i})},\mu^{-i,(k^{-i})}\Big).
	\end{equation}
	The rest of the algorithm stays the same.
	
	\subsection{Dynamic Level Model }
	
	In the previous section the agents were required to keep their rationality levels fixed, and hence the agents were
	not allowed to change their rationality levels during the game.
	To make the interactions between the agents more realistic, we now let the agents adapt their levels based on their observations.
	To this end,
	suppose the opponent's level $k^{-i}$ remains constant over time.
	We let agent $i$ adapt its rationality level based on its belief about the rationality level of its opponent.
	We also assume that the agent always plays one level higher than its opponent without, however,
	exceeding its maximum rationality level, $k^i_{\max}$.
	Let us pick the observation window to have length $w+1$ and let the current stage to be $t_N$.
Denote the observed trajectory from stage $t_{N-w}$ to stage $t_N$ as $s^{[N-w,N]} =\{s_{N-w},...,s_{N}\}$ and let the levels played by agent $i$ over these stages denoted as $k^{i,{[N-w,N]}} = \{k^i_{N-w}, \cdots, k^i_{N}\}$.
	The probability of observing the specific $s^{[N-w,N]}$ with a fixed $k^{-i}$ is
	\begin{equation}
	\mathbb{P}(s^{[N-w,N]} |k^{i,{[N-w,N]}},k^{-i}) = \prod_{n=N-w}^{N-1}P_h\Big(s_{n+1}|s_n,~\mu^{i,(k_n^{i})},\mu^{-i,(k^{-i})}\Big).
	\end{equation}
	Then
	\begin{equation}
	\hat{k}_N^{-i} \in
	\argmax_{k^{-i} \in \mathcal{K}^{-i,(k^i_{\max}-1)}}
	\mathbb{P}\Big(s^{[N-w,N]}|k^{-i},k^{i,{[N-w,N]}}\Big),
	\end{equation}
	can be selected as the maximum likelihood estimator of opponent's level at stage $t_N$. At the next stage $t_{N+1}$, agent $i$ would then play at level $\min \{\hat{k}^{-i}_N+1, k^i_{\max}\}$.
	\vspace{5mm}
	
	\section{Numerical Example}\label{sec:example}
	
	In this section, we illustrate the theory via a numerical example.
	We consider a two-player pursuit-evasion game on a 18$\times$18 discretized world, namely the lattice indexed as $L_h = \{1,2,...,18\}^4$.
	The mean velocity of the wind is generated randomly, while the wind covariance is set to $\sigma_{w}=0.4$ with no spatial correlation.
	The starting positions of the agents, the evasion states and the obstacles are shown in Fig.~\ref{fig:continuous_grid_def}.
	We use the MCAM to discretize the state space to obtain the grid shown in
	Fig.~\ref{fig:continuous_grid_def}.
	The terminal reward is assigned according to the function $G_h$ defined in (\ref{eqn:continuous_reward}).
	 Both agents have the same speed $v^1 = v^2 = 1$ and the same action set $\Theta^1 = \Theta^2 = \{0,\frac{\pi}{2}, \pi, \frac{3\pi}{2}\}$.
	
	\begin{figure}[H]
		\centering
	\begin{subfigure}[t]{0.49\linewidth}
	\hspace*{-4mm}
		\includegraphics[width=1.05\linewidth]{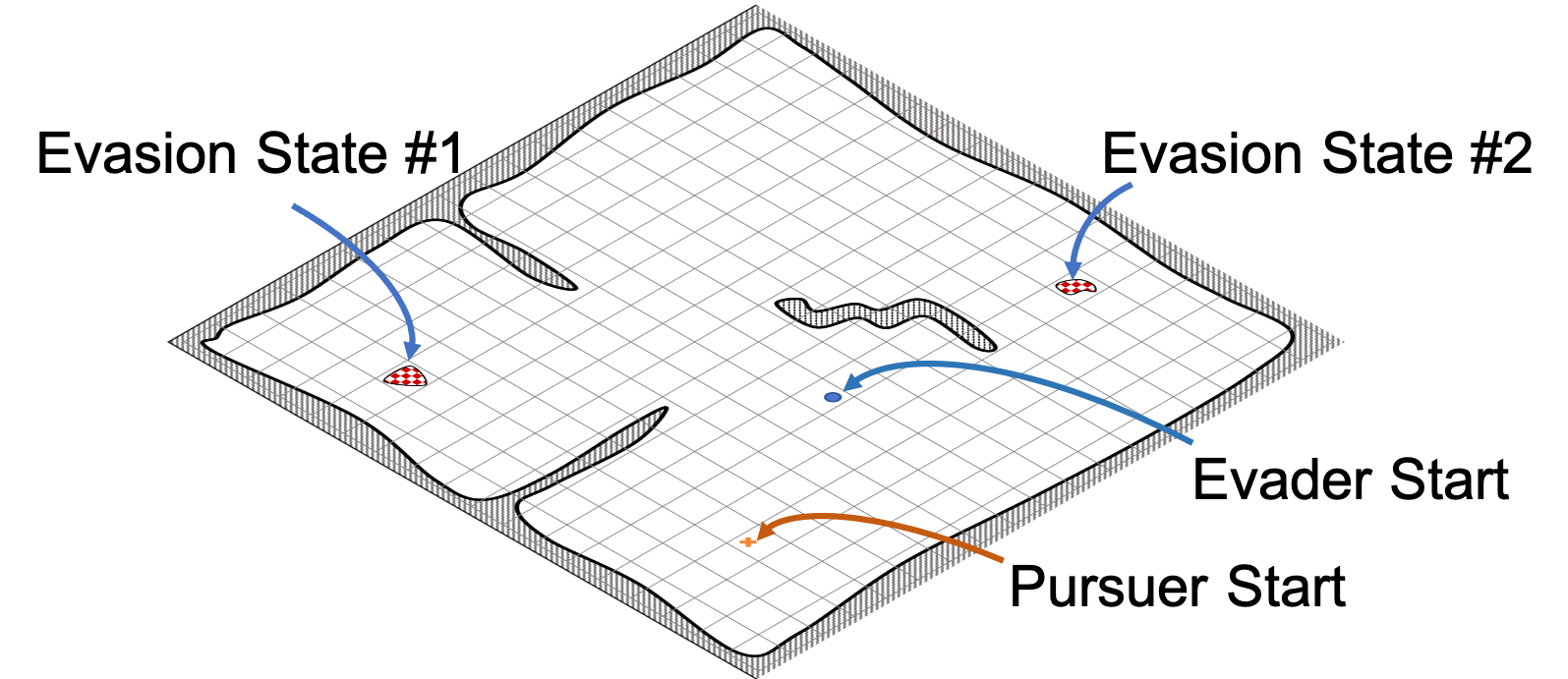}
	\end{subfigure}
        \begin{subfigure}[t]{0.49\linewidth}
			\includegraphics[width=1.05\linewidth]{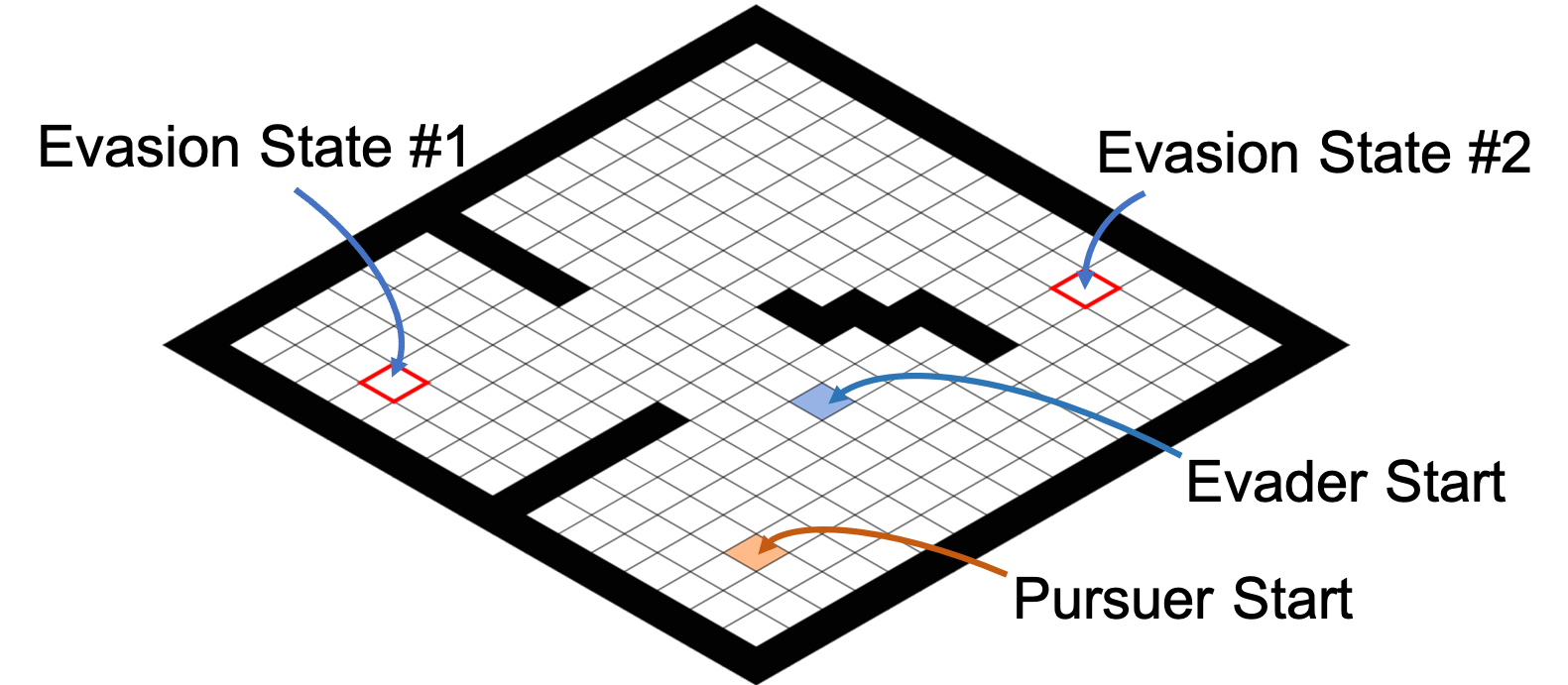}
	\end{subfigure}
		\caption{Continuous (left) and grid (right) space representations: starting positions and evade states. }
		\label{fig:continuous_grid_def}
	\end{figure}
	
	We assume that the level-0 agents pick their headings such that they will not cause an immediate crash, while choosing their action uniformly randomly.
	We construct the level-$k$ thinking hierarchy and calculate the policies of both agents for different rationality levels.
	We then simulate the resulting policies at some selected level pairs over 1500 games.
	To illustrate the results, we use two tables to present the win percentages of the selected scenarios.
	
	In Tables~\ref{tab:l2_evader} and \ref{tab:l2_pursuer}, we fix one agent to be level-2 while varying the level of the other agent.
	In both tables, the highest winning rates are attained at level 3.
	This result is expected, since level 3 policies are by definition the best responses to the level-2 policies.
	It can also be observed that after level 5 the performance only varies slightly, since the policies at high levels do not differ much from each other.
	\begin{table}[H]
		\caption{Pursuer win percentages against level-2 evader. Results are for 1500 simulated games.}\label{tab:l2_evader}
		\begin{center}
			\begin{tabular}{c c c c c c c}
				\hline
				$k^2 = 2$									&$k^1 = 1$ 	& $k^1 = 2$ & $k^1 = 3$ & $k^1 = 4$	& $k^1 = 5$	& $k^1 = 6$ \\
				\hline
				Pursuer Wins						&  48.5			  & 47.6		& 51.7			&50.9		& 51.2		& 51.1	\\
				~~ Due to Capture				& 36.4			 &38.2			& 45.3			&43.1		& 42.7		& 42.9	 		\\
			\end{tabular}
		\end{center}
		
		\end{table}
	
	\begin{table}[H]
		\caption{Evader win percentages against level-2 pursuer. Results are for 1500 simulated games.}	\label{tab:l2_pursuer}
		\begin{center}
			\begin{tabular}{c c c c c c c}
				\hline
				$k^1 = 2$						&$k^2 = 1$ 	& $k^2 = 2$ & $k^2= 3$ & $k^2= 4$ & $k^2= 5$ & $k^2= 6$	\\
				\hline
				Evader Wins						&  47.5			  & 52.4		& 53.8			&52.9	&53.2	&53.1			\\
				~~ Due to Evasion			& 35.6			 &42.2			& 45.3			&44.5	&44.7	&44.2 		\\
			\end{tabular}
		\end{center}
	
	\end{table}
	In Fig.~\ref{fig:trajectory} two sample trajectories are shown.
	The first trajectory depicts a level-3 evader against a level-2 pursuer.
	One observes the ``deceiving" behavior of the evader: it first moves towards evasion state \#2 (on the right), which tricks the level-2 pursuer to also go right and take the shorter route beneath the obstacles for defending the evade state \#2.
	The evader then suddenly turns left and goes to the evade state \#1 (on the left). When the evader reveals its true intention of evading at evasion state \#1, it is too late for the pursuer to capture the evader.
	The second trajectory depicts a successful capture.
	Different from the previous scenario, the level-3 pursuer has perfect information about the policies of the level-2 evader, as a result it can predict well what the evader will do in the next time step. The pursuer then simply acts accordingly to maximize the probability of a successful capture.
	It is also noteworthy that the stochastic wind field does disrupt the trajectories of the agents.
	For example, at point A in Fig.~\ref{fig:trajectory}(a), the evader intends to go left directly to the evasion state but is instead pushed upward due to the stochastic wind field.
	
	\begin{figure}[H]
		\centering
		\begin{subfigure}[t]{0.43\linewidth}
			\includegraphics[width=\linewidth]{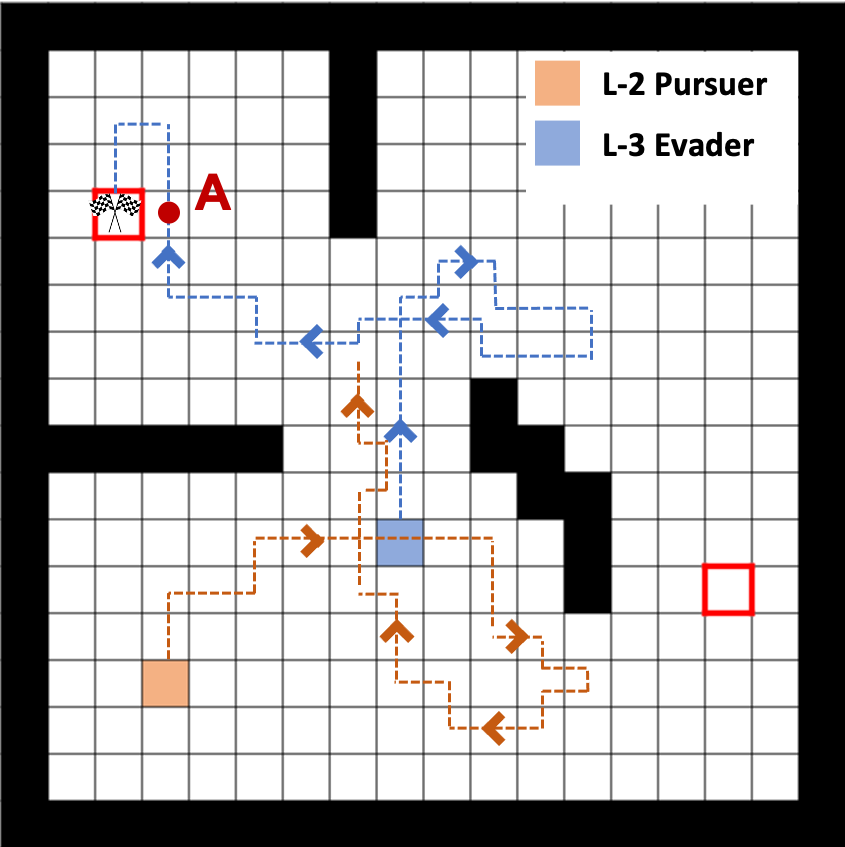}
			\caption{}
		\end{subfigure}\hspace{3mm}
		\begin{subfigure}[t]{0.43\linewidth}
			\includegraphics[width=\linewidth]{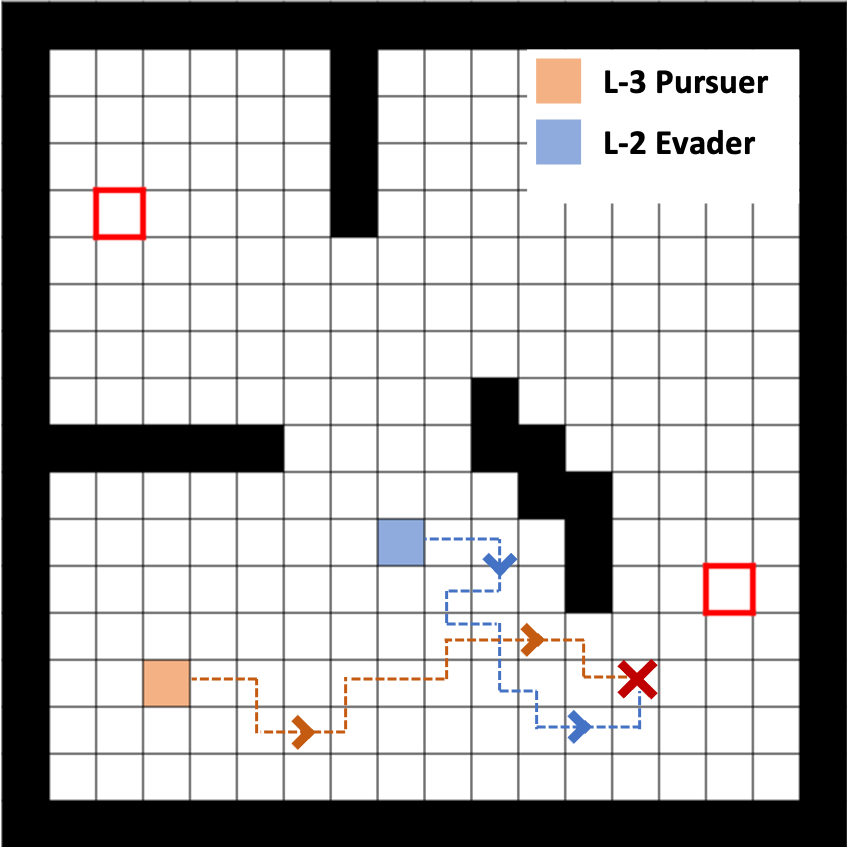}
			\caption{}
		\end{subfigure}
		\caption{Sample trajectories of agents playing fixed rationality levels. The first game is won by the evader; the second is won by the pursuer. (a) Level-2 Pursuer vs. level-3 Evader; (b) Level-3 Pursuer vs. level-2 Evader.}
		\label{fig:trajectory}
	\end{figure}
	
	The following color maps present an example of the outcomes from the static inferring algorithm introduced in Section~\ref{sec:infer}.
	The vertical axis represents the rationality levels of the opponents, and the $x$-axis represents the time instances.
	The color depicts the (normalized) conditional probability of the corresponding level.
	
	\begin{figure}[H]
		\centering
		\vspace{-5mm}
		\begin{subfigure}[t]{0.7\linewidth}
			\includegraphics[width=\linewidth]{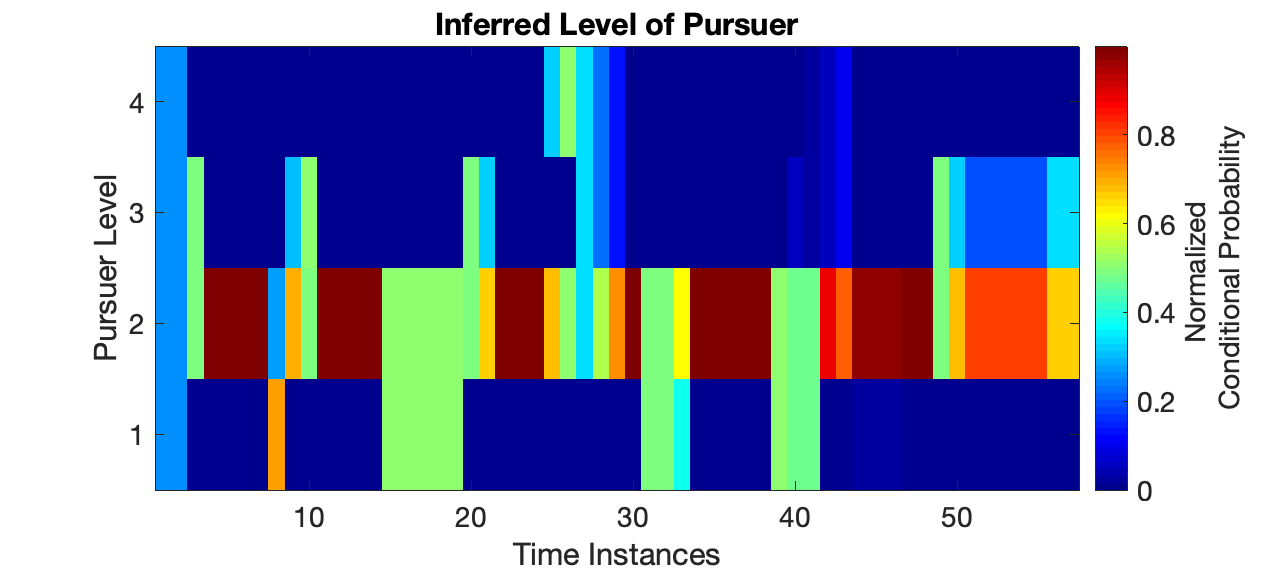}
			\caption{}
		\end{subfigure}
	
		\begin{subfigure}[t]{0.7\linewidth}
			\includegraphics[width=\linewidth]{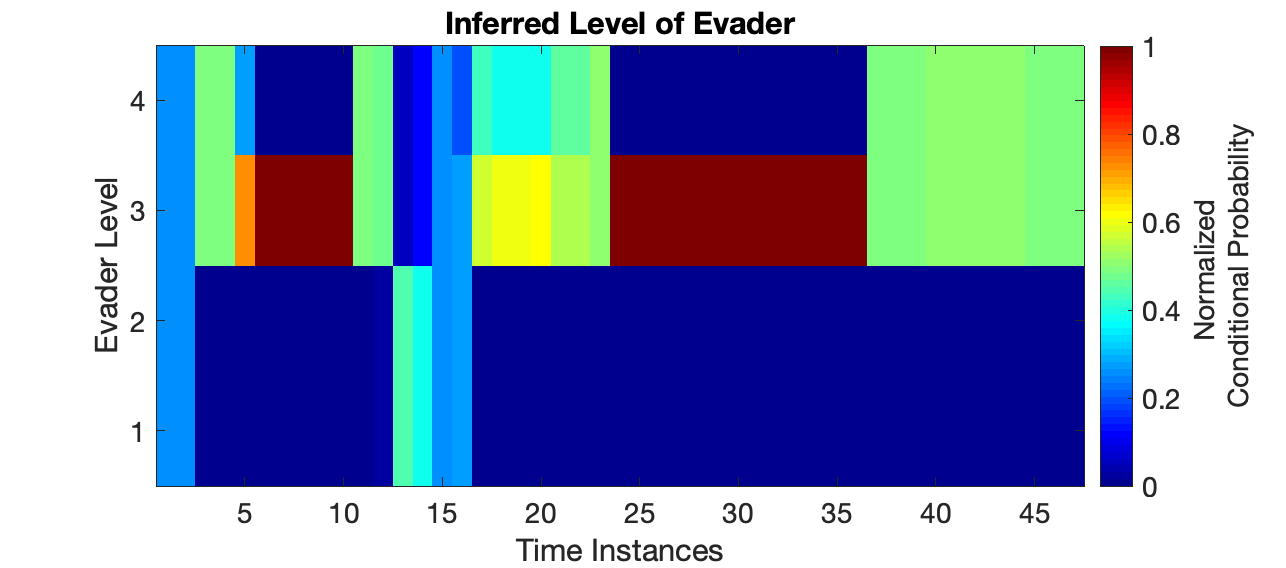}
			\caption{}
		\end{subfigure}
			\vspace{-2mm}
		\caption{Examples of inference results. 
		(a) A level-3 Evader's inference of a level-2 Pursuer; (b) A level-2 Pursuer's inference of a level-3 Evader.}
		\label{fig:inference}
		\vspace{-5mm}
	\end{figure}
	
	In Fig.~\ref{fig:inference}(a), the level-3 evader can infer the  rationality level of the level-2 pursuer with certain confidence.
	However, in Fig.~\ref{fig:inference}(b), the level-2 pursuer has  some trouble inferring the level-3 evader accurately.
	As discussed earlier, at high levels the policies of the agents are almost the same at most states, which means, for example, that
	a level-4 pursuer and a level-3 pursuer may take the same action in certain regions of the state space.
	In Fig.~\ref{fig:inference}(b), from time instances 37 to 49, the two agents may have entered such a region, where the policies of the level-3 pursuer and the level-4 pursuer are the same.
	This phenomenon makes the inferring process challenging at high rationality levels, in general.
	However, since the policies of both the pursuer and the evader become similar at high levels, even picking a wrong level does not harm performance significantly. %
Finally, we present in Fig.~\ref{fig:Dynamic Levels} the outcome of the dynamic level model introduced in Section~\ref{sec:infer}, in which each agent always plays one level higher than the inferred level of its opponent without exceeding its maximum rationality level.
One may notice some significant oscillations at the beginning of the game, but eventually, the pursuer starts to play at its highest level, and the evader plays at one level higher accordingly till the game terminates.

	\begin{figure}[H]
		\vspace{-1mm}
		\centering
		\includegraphics[width=0.9\linewidth]{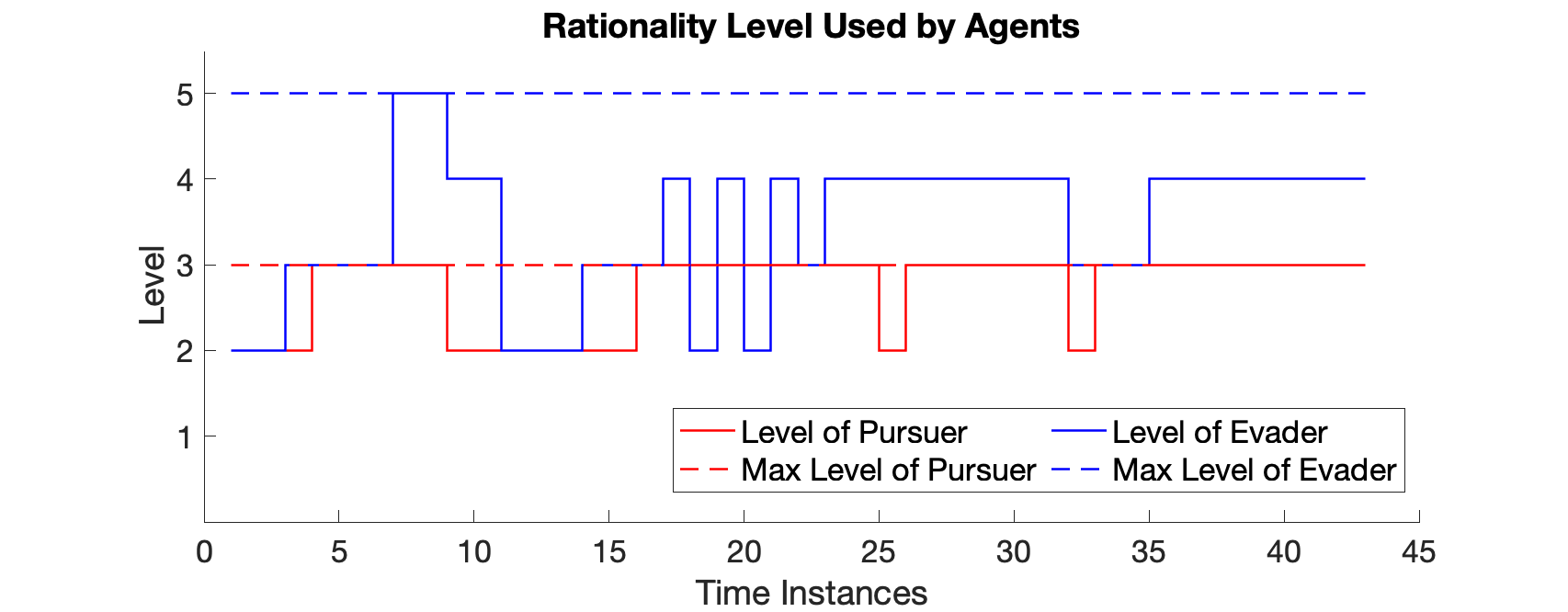}
		\caption{Example of dynamic level model. Evader has maximum rationality level of 3 and Pursuer has 5.}
		\label{fig:Dynamic Levels}
		\vspace{-5mm}
	\end{figure}
	
	\section{Conclusions}\label{sec:conclusion}
	
	In this work we have studied a classical pursuit evasion game in a stochastic environment
	where the agents are bounded rational.
	We provide a framework to incorporate the notion of bounded rationality in decision-making  for stochastic pursuit-evasion problems
	through Markov chain approximation followed by a $k$-level thinking framework.
	Bounded rational agents do not compute the Nash equilibrium,
	but they rather
	compute their level-$k$ strategy based on an iterative algorithm, which constructs the level-$k$ strategy from level-$(k-1)$ strategy.
	A bounded-rational agent with maximum rationality  level $k_{\max}$ can employ a strategy of any level $k=0,1,\cdots,k_{\max}$.
	However, the optimal rationality level $k^{*}$ should be determined by the opponent's corresponding
	rationality level, and hence, it is crucial for an agent to be able to infer its opponent's rationality level.
	We propose two algorithms (fixed and dynamic) that assist the agent with inferring its opponent's rationality level.
	Through simulation results we have demonstrated the behavioral and statistical outcomes of the game.
	In particular, we have demonstrated how an agent playing at a higher rationality level may be able to deceive another agent of lower rationality level.
	As part of future work, it may be of interest to investigate this framework in a team game setup and study the implications of different rationality levels on the overall performance of the team.

	
	\bibliographystyle{ieeetr}
	\bibliography{Ref}

\end{document}